\documentclass[sigconf]{acmart}
\usepackage{threeparttable}
\usepackage{color, xspace}
\usepackage{booktabs} 
\usepackage{amsmath}
\usepackage{amsfonts}
\usepackage{comment}
\usepackage{footmisc}
\usepackage{mathrsfs}
\usepackage{graphicx}
\usepackage{subfigure}
\usepackage{multirow}
\usepackage{algorithm}
\usepackage{algorithmic}
\usepackage{multicol}  
\usepackage{enumitem}
\usepackage{array}
\usepackage{float}
\usepackage{bm}
\newtheorem{assumption}{Assumption}

\usepackage{soul, color, xcolor,xspace}

\copyrightyear{2024}
\acmYear{2024}
\setcopyright{acmlicensed}\acmConference[CIKM '24]{Proceedings of the 33rd ACM International Conference on Information and Knowledge Management}{October 21--25, 2024}{Boise, ID, USA}
\acmBooktitle{Proceedings of the 33rd ACM International Conference on Information and Knowledge Management (CIKM '24), October 21--25, 2024, Boise, ID, USA}
\acmDOI{10.1145/3627673.3679682}
\acmISBN{979-8-4007-0436-9/24/10}

\begin{document}

\title{
Efficient and Robust Regularized Federated Recommendation 
}

\author{Langming Liu}
\affiliation{%
  \institution{City University of Hong Kong}
  \city{Hong Kong}
  \country{China}
}

\author{Wanyu Wang}
\affiliation{%
  \institution{City University of Hong Kong}
  \city{Hong Kong}
  \country{China}
}

\author{Xiangyu Zhao}
\affiliation{%
  \institution{City University of Hong Kong}
  \city{Hong Kong}
  \country{China}
}

\author{Zijian Zhang}
\authornote{Zijian Zhang is corresponding author. zhangzj2114@mails.jlu.edu.cn}
\affiliation{%
  \institution{Jilin University}
  \city{Changchun}
  \country{China}
}

\author{Chunxu Zhang}
\affiliation{%
  \institution{Jilin University}
  \city{Changchun}
  \country{China}
}

\author{Shanru Lin}
\affiliation{%
  \institution{City University of Hong Kong}
  \city{Hong Kong}
  \country{China}
}

\author{Yiqi Wang}
\affiliation{%
  \institution{Michigan State University}
  \city{East Lansing}
  \country{United States}
}

\author{Lixin Zou}
\affiliation{%
  \institution{Wuhan University}
  \city{Wuhan}
  \country{China}
}

\author{Zitao Liu}
\affiliation{%
  \institution{Jinan University}
  \city{Guangzhou}
  \country{China}
}

\author{Xuetao Wei}
\affiliation{%
  \institution{Southern University of Science and Technology}
  \city{Shenzhen}
  \country{China}
}

\author{Hongzhi Yin}
\affiliation{%
  \institution{The University of Queensland}
  \city{Brisbane}
  \country{Australia}
}

\author{Qing Li}
\affiliation{%
  \institution{The Hong Kong Polytechnic University}
  \city{Hong Kong}
  \country{China}
}

\renewcommand{\shortauthors}{Langming Liu, et al.}
\begin{abstract}

Recommender systems play a pivotal role across practical scenarios, showcasing remarkable capabilities in user preference modeling. 
However, the centralized learning paradigm predominantly used raises serious privacy concerns.
The federated recommender system (FedRS) addresses this by updating models on clients, while a central server orchestrates training without accessing private data.
Existing FedRS approaches, however, face unresolved challenges, including non-convex optimization, vulnerability, potential privacy leakage risk, and communication inefficiency.
This paper addresses these challenges by reformulating the federated recommendation problem as a convex optimization issue, ensuring convergence to the global optimum.
Based on this, we devise a novel method, RFRec, to tackle this optimization problem efficiently. 
In addition, we propose RFRecF, a highly efficient version that incorporates non-uniform stochastic gradient descent to improve communication efficiency. 
In user preference modeling, both methods learn local and global models, collaboratively learning users' common and personalized interests under the federated learning setting. 
Moreover, both methods significantly enhance communication efficiency, robustness, and privacy protection, with theoretical support. 
Comprehensive evaluations on four benchmark datasets demonstrate RFRec and RFRecF's superior performance compared to diverse baselines. 
The code is available to ease reproducibility\footnote{\url{https://github.com/Applied-Machine-Learning-Lab/RFRec}}.

\end{abstract}



\maketitle

\section{Introduction}

In recent years, recommender systems (RSs)~\cite{resnick1997recommender,jannach2010recommender,naumov2019deep,hu2008collaborative,koren2009matrix,liu2023linrec,lin2022adafs,zhang2023denoising,wang2023multi,li2022gromov,li2023e4srec} have become prevalent in various practical scenarios, such as e-commerce, online movie or music platforms, providing personalized recommendations for users~\cite{lu2015recommender,schafer1999recommender,celma2010music,liu2023multi,li2023automlp,liu2023exploration,lin2023autodenoise}. Matrix Factorization (MF)-based models~\cite{hu2008collaborative,koren2009matrix,koren2008factorization,mnih2007probabilistic}, the most prevalent RS solutions, leverage a centralized learning paradigm that collects user data and trains the model in a central server. The centralized approach can address user preferences well and has achieved great success. 
However, there emerge increasing concerns from individual users about privacy issues during information storage and collection in central server~\cite{yang2020federated,yang2019federated}. 
In response, the authorities are strengthening data privacy protection legally~\cite{yang2019federated}, such as issuing the General Data Protection Regulation (GDPR)~\cite{regulation2018general}.
Exploring a practical approach that can train the models without directly utilizing user data is urgent in the recommendation domain.

Federated Learning (FL)~\cite{konevcny2016federated,konevcny2016federated1,mcmahan2016federated} has proven its efficacy in addressing the aforementioned data-isolated problem. In contrast to the centralized learning paradigm, FL learns the optimal model by continuously communicating the model parameters or gradients between the local client (user) and central server without aggregating client data to the server. Therefore, FL is promising to be the solution to address the privacy issue in the recommendation.    
Recently, some efforts have been exerted to integrate the FL technique into recommender systems to address the privacy issue, named as federated recommender systems (FedRS)~\cite{ammad2019federated,lin2020fedrec,chai2020secure,muhammad2020fedfast,lin2020meta,wu2021fedgnn,dolui2019towards,sun2022survey}. Using the same optimization formulation as centralized RS, mainstream FedRS methods (e.g., MF-based FedRS) learn the user and item latent feature vectors and utilize their inner product to predict the rating matrix. 
In the FL setting, the user feature vectors are updated only on the clients. In contrast, the gradients of item feature vectors are aggregated to the server for item feature matrix updating. This approach can train the recommendation model without accessing and utilizing user data, alleviating privacy concerns. 

However, we recognize that existing FedRS methods generally suffer from several deficiencies. 
\textbf{1) Non-convex Optimization}: most FedRS methods solve the RS optimization problem formulated as the MF error minimization directly~\cite{ammad2019federated,lin2020fedrec,chai2020secure}, which is essentially non-convex. The non-convex optimization may converge to sub-optimal points, jeopardizing the recommendation performance. Moreover, constrained by the FL setting, the user and item feature vectors defined in the formulation must be updated in the client and server, respectively, which can lead to unstable model training.
\textbf{2) Vulnerability}: the existing methods leverage the alternating update rule, where the client and server highly depend on each other to update~\cite{ammad2019federated,wu2021fedgnn,takacs2012alternating}. Therefore, the update process may terminate when some clients cannot participate due to network or device issues.  
\textbf{3) Privacy Leakage Risk}: in the training process, the transmission of gradients may still lead to private information leaks~\cite{chai2020secure}. Specifically, knowing two continuous gradients of a user during the communication process, the server can deduce the user ratings by the algebra operation of these two gradients.
Though privacy protection methods can be equipped to prevent privacy leakage, such as
leveraging homomorphic encryption and generating pseudo items~\cite{liu2023privaterec,chai2020secure,minto2021stronger},
such methods bring considerable additional computational or communication costs.
\textbf{4) Communication Inefficiency}: FedRS methods have a low convergence rate due to the alternating update rule, no matter alternating least squares (ALS)~\cite{takacs2012alternating,haldar2009rank} or alternating stochastic gradient descent (SGD)~\cite{bottou2010large,koren2009matrix}.
Specifically, these methods converges to the optimal with $\varepsilon$-accuracy in at least $O(\log m + \log\frac{1}{\varepsilon})$ iterations~\cite{lee2023randomly,jain2013low,takacs2012alternating}, where $m$ is the number of items.  Hence, the number of communication rounds of existing methods is also $O(\log m + \log\frac{1}{\varepsilon})$, which is tremendous when $m$ is large.

In this paper, we first propose a problem formulation for FedRS and reformulate its optimization as a convex problem, which guarantees that the model converges to global optimal.
As a regularized empirical risk minimization (RERM)~\cite{marteau2019beyond, zhang2017stochastic}, it consists of cumulative local loss (task term) and global loss (regularization term). 
Under the guidance of this formulation, we design two methods, namely RFRec and its faster version, RFRecF.
In particular, we propose a \textbf{R}egularized \textbf{F}ederated learning method for \textbf{Rec}ommendation (\textbf{RFRec}) to solve this RERM problem in a local GD manner. It theoretically ensures strict convergence to optimal with a linear convergence rate, which has an edge on communication efficiency.
In addition, to pursue further communication efficiency, we leverage a non-uniform SGD training manner to update the task term and regularization term stochastically (\textbf{RFRecF}).
This optimization mechanism enables RFRecF to achieve fewer expected communication rounds per iteration, thereby reducing communication costs and striking a balance between performance and efficiency.
The minimum communication iterations of both methods is $O(\log\frac{1}{\varepsilon})$, which is independent of item number $m$ and achieves a remarkable reduction over existing methods.
Our main contributions can be summarized as follows:

\begin{itemize}[leftmargin=*]
\item We propose a convex optimization formulation of FedRS, formulated as a regularized empirical risk minimization (RERM) problem, which can guide FedRS methods design with theoretical support for convergence and communication efficiency.   
\item We put forward RFRec to solve the proposed RERM problem and the faster version, RFRecF, for further communication efficiency.
We theoretically and experimentally prove their robustness and privacy preservation capability, and they own a guarantee of advancing communication efficiency.
\item Extensive experiments are conducted on four public benchmark datasets, demonstrating the state-of-the-art performance of our proposed methods against advanced baselines.
We also conduct in-depth analysis to verify the efficacy and efficiency.
\end{itemize}


\section{Preliminary}
\subsection{Recommender Systems}
Denote the number of users and items as $n,m$, respectively.  
In RSs, the prevalent MF-based methods~\cite{koren2009matrix,koren2008factorization,mnih2007probabilistic} suppose the rating matrix $\hat{\boldsymbol{R}}\in\mathbb{R}^{n\times m}$ is the inner product of user feature matrix $\boldsymbol{U}\in\mathbb{R}^{d\times n}$ and item feature matrix $\boldsymbol{V}\in\mathbb{R}^{d\times m}$, formulated as
$\hat{\boldsymbol{R}} = \boldsymbol{U}^T \boldsymbol{V}.$
So the predicted rating of user $i$ to item $j$ is computed as $\hat{r}_{ij} = \boldsymbol{u}_i^T \boldsymbol{v}_j$, where $\boldsymbol{u}_i\in\mathbb{R}^d$ and $\boldsymbol{v}_j\in\mathbb{R}^d$ are user and item latent factor vectors, respectively. 
The optimization objective is as follows:
\begin{equation}
\label{MF_opt}
   \begin{aligned}
    J &= \Vert \boldsymbol{R} - \boldsymbol{U}^T \boldsymbol{V}\Vert^2 + \lambda_u\Vert \boldsymbol{U}\Vert^2 + \lambda_v\Vert \boldsymbol{V}\Vert^2\\
    &= \sum_{i=1}^n\sum_{j=1}^m({r}_{ij} - \boldsymbol{u}_i^T \boldsymbol{v}_j)^2 + \lambda_u\sum_{i=1}^n\Vert \boldsymbol{u}_i\Vert^2 + \lambda_v\sum_{j=1}^m\Vert \boldsymbol{v}_j\Vert^2,
    \end{aligned}
\end{equation}
where $\lambda_u, \lambda_v$ are scaling parameters of penalties. 
\subsection{Federated Learning}
The prevalent optimization objective of FL~\cite{li2020federated} is an ERM problem 
\begin{equation}
\min\limits_{x}\frac{1}{n}\sum_{i=1}^n f_i(x),
\label{FL}
\end{equation}
where $n$ is the number of clients, $x$ encodes the parameters of the global model (e.g., weight and bias of NN), and
$f_i$ represents the aggregate loss of the $i$-th client. 


\subsection{Federated Recommender Systems}
\label{sec:FedRS}
The optimization objective of FedRS is the same as RS. However, in the setting of FL, the rating vector $\boldsymbol{R}_i\in\mathbb{R}^{1\times d}$ is only available on user $i$, which obstacles the update of item vectors. The solution is to separate updates into two parts: clients and servers. 

The user vectors $\boldsymbol{u}_i$ are updated locally by GD:
\begin{equation}
\label{eq:grad}
\nabla_{u_i} J= -2\sum_{j=1}^m({r}_{ij} - \boldsymbol{u}_i^T \boldsymbol{v}_j)\boldsymbol{v}_j + 2\lambda_u \boldsymbol{u}_i
\end{equation}
The item vector $\boldsymbol{v}_j$ is updated in central server by aggregating gradient $h(i,j)$ calculated in clients:
\begin{equation}
\label{eq:alter}
    \begin{aligned}
\text{client: }   &h(i,j) = ({r}_{ij} - \boldsymbol{u}_i^T \boldsymbol{v}_j)\boldsymbol{u}_i\\
\text{server: }  &\nabla_{v_j} J = -2\sum_{i=1}^n h(i,j) + 2\lambda_v \boldsymbol{v}_j
    \end{aligned}
\end{equation}
The updated item vectors will be distributed to clients for the next update step. 
FedRS is implicitly meant to train several local models (i.e., $\boldsymbol{u}_i$) and a global model (i.e., $\boldsymbol{V}$). We observe a gap between the optimization objectives of FedRS and general FL, making it inconvenient to design effective algorithms and conduct analyses. 


\section{Methodology}
We reformulate the FedRS problem with an RERM formulation. Then, we propose RFRec and RFRecF to solve the problem, as illustrated in Figure~\ref{fig:Architecture}, providing robustness and strong privacy protection. The convergence and communication results are provided. 


\subsection{Problem Formulation of FedRS}
To interpret the recommendation task under the FL setting, we formulate FedRS as a RERM problem, learning the local and global models simultaneously. 
We define the local models as $x = (x_1, x_2, \cdots,x_n)$,
where tuple $x_i = (\boldsymbol{u}_i, \boldsymbol{V}_{(i)})$ denotes the local model of client (user) $i$, and $\boldsymbol{u}_i\in\mathbb{R}^{d}$ is the user feature vector. Notice that $\boldsymbol{V}_{(i)}\in\mathbb{R}^{d\times m}$ is the local item matrix, which differs from $v_j\in\mathbb{R}^{d}$. Denote $\boldsymbol{R}_i\in\mathbb{R}^{1\times m}$ as user $i$'s rating vector. The global model is $\bar{\boldsymbol{V}} = \frac{1}{n}\sum_{i=1}^n \boldsymbol{V}_{(i)}$.
Here, we formularize the problem of FedRS as follows
\begin{equation}
\min\limits_{x_1,\cdots,x_n} \{F(x):= f(x)+\lambda \psi(x)\}, 
\label{RFL_opt}
\end{equation}
$$f(x):=\sum_{i=1}^n f_i(x_i) = \sum_{i=1}^n\{\Vert\boldsymbol{R}_{i} - \boldsymbol{u}_i^T \boldsymbol{V}_{(i)}\Vert^2 + \lambda_u\Vert \boldsymbol{u}_i\Vert^2\},$$
$$\psi(x):=\sum_{i=1}^n \psi_i(x_i) = \frac{1}{2}\sum_{i=1}^n \Vert \boldsymbol{V}_{(i)}-\bar{\boldsymbol{V}}\Vert^2,$$
where $f$ and $\psi$ are the cumulative local loss and the regularization term.
$f_i$ and $\psi_i$ are the losses of client $i$. 
Notice that $\Vert\cdot\Vert$ refers to the L2-norm. 
Minimizing $f$ obtains optimal local model $x_i = (\boldsymbol{u}_i, \boldsymbol{V}_{(i)})$ for each client. 
The regularization term $\psi$ is added to penalize the difference between $\boldsymbol{V}_{(i)}$ and $\bar{\boldsymbol{V}}$. Therefore, minimizing $\psi$ obtains the global item matrix $\bar{\boldsymbol{V}}$. In addition, compared to the original optimization formulation~\eqref{MF_opt}, we omit the penalty term $\lambda_v\Vert \boldsymbol{V}_{(i)}\Vert^2$, since the regularization term $\psi$ can undertake the same role as it and this penalty term may jeopardize the training process. 

By tuning $\lambda$ to balance the learning of local and global models, we can learn the optimal recommendation model by solving problem~\eqref{RFL_opt}. 
In addition, the regularization term can provide strong convexity, referred to in Lemma~\ref{lemma:convex_2}.
Moreover, the formulation~\eqref{RFL_opt} can guide the designing of efficient algorithms.  
We denote $x(\lambda):=(x_1(\lambda),\cdots, x_n(\lambda))$ as the optimal model of problem~\eqref{RFL_opt}. Then, we aim at devising relative training methods to obtain $x(\lambda)$.


\subsection{Methods}
\label{sec:algo}
We propose two methods to solve optimization problem~\eqref{RFL_opt}, RFRec and RFRecF. 
Instead of communicating gradients like most FedRS methods, the proposed methods communicate models for federated learning, alleviating the risk of privacy leaks. In addition, they achieve higher communication efficiency than FedRS methods. 


\begin{figure}
\vspace{-2mm}
    \centering
    \includegraphics[width=1.0\linewidth]{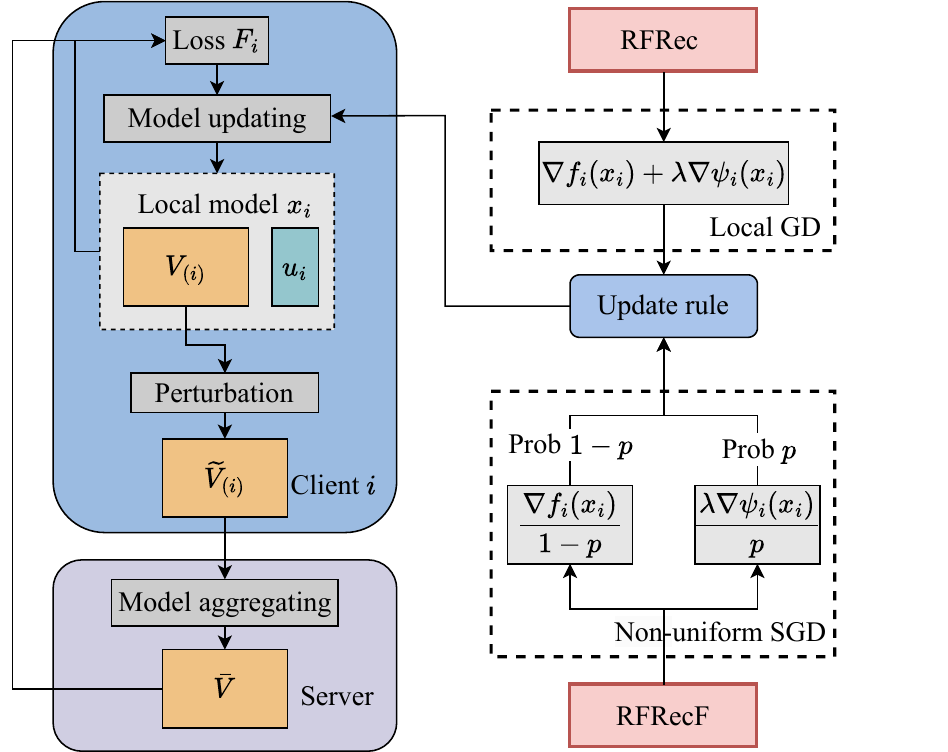}
    \caption{Overview of RFRec and RFRecF. 
    }
    \label{fig:Architecture}
    \vspace{-4mm}
\end{figure}

\subsubsection{\textbf{RFRec}}
To achieve the optimal solution of problem~\eqref{RFL_opt} stably and efficiently, we propose RFRec, which conducts local GD in clients to update models. First, We add the regularization term $\psi_i(x_i)$ to main target to get the local optimization target $F_i(x_i) = f_i(x_i)+\lambda\psi_i(x_i)$. 
Then, the local model $x_i = (\boldsymbol{u}_i, \boldsymbol{V}_{(i)})$ can be updated on client $i$ by gradient descent update rule as follows
\begin{equation}
    \label{update_RFRec}
    x_i^{k+1} = x_i^k - \alpha\nabla F_i(x_i),
\end{equation}
where $\alpha$ is the learning rate, $\lambda$ is the penalty parameter, and $x^k$ denotes the model at step $k$.
However, clients $i$ cannot compute $\psi_i(x_i)$ without knowing the average model $\bar{\boldsymbol{V}}$. So, the clients upload $\boldsymbol{V}_{(i)}$ to the server for aggregation (i.e., calculating $\bar{\boldsymbol{V}}$). Afterward, $\bar{\boldsymbol{V}}$ is distributed to clients for the local model update. 
The final outputs of RFRec are the local models $\{\boldsymbol{u}_i\}_{i=1}^n$ and the global model $\bar{\boldsymbol{V}}$, which means $\boldsymbol{u}_i$ is only available on user $i$, and $\bar{\boldsymbol{V}}$ is shared parameter. The pseudo-code is specified in Algorithm~\ref{algorithm1}. 

 
The significant advantages of RFRec are guaranteeing convergence to optimal and possessing a nice communication efficiency, specified in Section~\ref{sec:convergence} and~\ref{sec:communication}. However, the communication cost is still high as RFRec needs two communication rounds per iteration.


\begin{algorithm}
	\caption{RFRec}
	\label{algorithm1}
	\begin{algorithmic}[1]
		\STATE\textbf{Input:} $n$ clients and server, Ratings $\boldsymbol{R}$ 
        \STATE\textbf{Output:} local user vectors $\{\boldsymbol{u}_i\}_{i=1}^n$ and global item matrix $\bar{\boldsymbol{V}}$
		\STATE\textbf{Initialization:} $x = \{x_i\}_{i=1}^n = \{\boldsymbol{u}_i, \boldsymbol{V}_{(i)}\}_{i=1}^n$, $\bar{\boldsymbol{V}}\sim\mathcal{N}(0,10^{-4})$, parameters $\lambda,\alpha$ 
        \FOR{$k=0,1,\cdots, K$} 
        \STATE Local client update
            \FOR{$i=1,\cdots, n$ in parallel} 
            \STATE $x_i\leftarrow x_i-\alpha(\nabla f_i(x_i)+\lambda\nabla\psi_i(x_i))$
            \STATE Client $i$ generates and sends perturbed $\widetilde{\boldsymbol{V}}_{(i)}$ to server
            \ENDFOR
        \STATE Central server aggregate: $\bar{\boldsymbol{V}}\leftarrow \frac{1}{n}\sum_{i=1}^n \widetilde{\boldsymbol{V}}_{(i)}$
        \STATE Distributed $\bar{\boldsymbol{V}}$ to clients 
        \ENDFOR
	\end{algorithmic}  
\end{algorithm}
\subsubsection{\textbf{RFRecF}}
To further improve the communication efficiency based on RFRec, we apply a non-uniform SGD~\cite{bottou2010large,hanzely2020federated} method in the FL training process and name it RFRecF (fast version). We define the stochastic gradient of $F$ as follows:
\begin{equation}
\label{SG}
G(x):=\left\{
\begin{aligned}
&\frac{\nabla f(x)}{1-p}\quad&\text{with probability }&\quad&1-p,\\
&\frac{\lambda\nabla\psi(x)}{p}&\text{with probability }& &p,\quad
\end{aligned}\right.
\end{equation}
where $G(x)$ is an unbiased estimator of $\nabla F(x)$. Notice that the SGD here is not the standard approach that randomly samples data or mini-batch data for updates~\cite{bottou2010large}. Instead, we add randomness by stochastically using the gradient of the main loss $f$ or regularization term $\psi$.  
Then we similarly define the local stochastic gradient $G_i(x_i)$ as an unbiased estimator of $\nabla F_i(x_i)$. The updating step of RFRecF on client $i$ is formulated as 
\begin{equation}
x_i^{k+1}=x_i^k-\alpha G_i(x_i^k).
\label{update_RFRecF}
\end{equation}
The clients conduct the stochastic update throughout the process, while the central server only undertakes the model aggregation.
Specifically, we generate a Bernoulli random variable $\zeta = 1$ with probability $p$ and $0$ with probability $1-p$ each time. If $\zeta = 0$, the client conducts a local stochastic update, and else the server calculates the average model $\bar{\boldsymbol{V}}$. In addition, the local stochastic update is divided into two parts: \textbf{main update} (i.e., GD of $\nabla f_i(x_i)$) with probability $1-p$ and \textbf{moving to average} (i.e., GD of $\nabla \psi_i(x_i)$) with probability $p$. Note that the communication is only needed when the two consecutive $\zeta$ are different (i.e., $\zeta_{k-1} = 0, \zeta_{k} = 1$ or $\zeta_{k-1} = 1, \zeta_{k} = 0$), which provide high communication efficiency. The pseudo-code of RFRecF is specified in Algorithm~\ref{algorithm2}. 

RFRecF can provide even higher communication efficiency than RFRec, and the theoretical support is in Section~\ref{sec:communication}. Note that both proposed methods conduct complete model updates at clients, while the server only aggregates models. This approach is critical to enhance \textbf{Robustness}, and we will discuss it in the next section.  


\begin{algorithm}
	\caption{RFRecF}
	\label{algorithm2}
	\begin{algorithmic}[1]
		\STATE\textbf{Input:} $n$ clients and server, Ratings $\boldsymbol{R}$ 
        \STATE\textbf{Output:} local user vectors $\{\boldsymbol{u}_i\}_{i=1}^n$ and global item matrix $\bar{\boldsymbol{V}}$
		\STATE\textbf{Initialization:} $x = \{x_i\}_{i=1}^n = \{\boldsymbol{u}_i, \boldsymbol{V}_{(i)}\}_{i=1}^n$, $\bar{\boldsymbol{V}}\sim\mathcal{N}(0,10^{-4})$, parameters $\lambda,\alpha,p$
        \FOR{$k=0,1,\cdots, K$} 
            \STATE random variable $\zeta_k = 1$ with prob $p$ and $0$ with prob $1-p$
            \IF{$\zeta_k = 0$}
                \STATE local client update
                \FOR{$i=1,\cdots, n$ in parallel} 
                \IF{$\zeta_{k-1} = 0$ or $k=0$} 
                \STATE \textbf{main update}: $x_i\leftarrow x_i-\frac{\alpha}{1-p}\nabla f_i(x_i)$
                \ELSE
                \STATE \textbf{moving to average}: $\boldsymbol{V}_{(i)}\leftarrow \boldsymbol{V}_{(i)}-\frac{\alpha}{p}\lambda(\boldsymbol{V}_{(i)}-\bar{\boldsymbol{V}})$
                \ENDIF
                \ENDFOR
            \ELSE
                \STATE central server aggregate: $\bar{\boldsymbol{V}}\leftarrow \frac{1}{n}\sum_{i=1}^n \widetilde{\boldsymbol{V}}_{(i)}$
            \ENDIF
            \STATE \textit{Communication round}
            \IF{$\zeta_{k-1} = 0, \zeta_k = 1$}
                \FOR{$i=1,\cdots, n$ in parallel} 
                \STATE client $i$ generates and sends          perturbed $\widetilde{\boldsymbol{V}}_{(i)}$ to server 
                \ENDFOR
            \ELSIF{$\zeta_{k-1} = 1, \zeta_k = 0$}
                \STATE server send $\bar{\boldsymbol{V}}$ to clients
            \ENDIF
        \ENDFOR
	\end{algorithmic}  
\end{algorithm}

\subsection{In-depth Analysis of Robustness}
The existing methods are vulnerable when facing low participation of clients (e.g., some devices lose connection) since client and server updates highly depend on each other (Section~\ref{sec:FedRS}). 
To address the issue, we strip off update tasks from the server, which means the server only undertakes aggregation. There are two advantages: 
\textbf{1) For server}: the model aggregation of $\boldsymbol{V}_{(i)}$ will not be influenced when a portion of devices cannot participate in aggregation. By the law of large numbers, the sample average converges almost surely to the expected value, so we can use the average of remaining devices $\frac{1}{n_t}\sum_{i=1}^{n_t} \boldsymbol{V}_{(i)}$ to approximate $\frac{1}{n}\sum_{i=1}^n \boldsymbol{V} _{(i)}$, where $n_t<n$ represents the number of participated devices at round $t$.
\textbf{2) For clients}: each client can continue local updates even without connection to the server for a while since the average model $\bar{\boldsymbol{V}}$ is stable.   

\subsection{Enhanced Privacy Protection}
Initially, there was a vast difference between local model $\boldsymbol{V}_{(i)}$, where the information leak is more likely to happen when each client uploads $\boldsymbol{V}_{(i)}$ to server.
To address this potential issue, we apply a clipping mechanism and leverage local differential
privacy (LDP)~\cite{arachchige2019local,choi2018guaranteeing} (e.g., using the Laplace mechanism) on $\boldsymbol{V}_{(i)}^k$ as follows:
\begin{equation}
\label{eq:perturbation}
    \widetilde{\boldsymbol{V}}_{(i)}^k = \text{Clip}(\boldsymbol{V}_{(i)}^k,\delta) + \text{Laplace}(0,s)
\end{equation}
The protecting mechanism ensure $\epsilon$-DP ($\epsilon = \frac{2\delta}{s}$), where $\epsilon$ is privacy budget, smaller the better. Then, the perturbed $\widetilde{\boldsymbol{V}}_{(i)}^k$ is sent to the server for averaging. The perturbation can be mostly offset when the server takes the average of $\widetilde{\boldsymbol{V}}_{(i)}^k$, which has a slight influence on convergence. In Section~\ref{sec:convergence}, we will show that RFRec and RFRecF's convergence results still hold even after adding this perturbation.

\subsection{Convergence Results}
\label{sec:convergence}
In this section, we give the convergence results of the proposed methods, demonstrating their effectiveness theoretically and laying the foundation for deducing communication results, and the important proofs are attached in \textbf{Appendix}~\ref{appendix}. 
We first reasonably assume the model parameters are uniformly bounded.
\begin{assumption} 
\label{assum}
Suppose the model parameters are uniformly bounded such that $\Vert \boldsymbol{u}_i\Vert\le M_u, \Vert \boldsymbol{V}_{(i)}\Vert\le M_v,\Vert \boldsymbol{R}_i\Vert\le M_r$. 
\end{assumption}
The theoretical results in this paper are all based on Assumption~\ref{assum}, and we omit the description of assumption for simplicity. 
Then, we define convexity and smoothness as follows.
\begin{definition}
A function $g:\mathbb{R}^d\to \mathbb{R}$ is
$L-$smooth with $L>0$ if
$$g(x)\le g(y)+\nabla g(y)^T(x-y)+\frac{L}{2}\Vert x-y\Vert^2, \quad\forall x,y\in \mathbb{R}^d,$$
and it is $\mu$-strongly convex with $\mu >0$ if
$$g(x)\ge g(y)+\nabla g(y)^T(x-y)+\frac{\mu}{2}\Vert x-y\Vert^2 ,\quad\forall x,y\in \mathbb{R}^d.$$
\end{definition}

\begin{lemma}
\label{lemma:convex_1}
Each $f_i$ ($i = 1,\cdots, n$) is $L$-smooth.
\end{lemma}


\begin{lemma}
\label{lemma:convex_2} Let $\lambda>\frac{2}{\lambda_u}M_r^2+6M_u^2$, we have 
$F$ is $\mu$-strongly convex and $L_F$-smooth with $L_F = L+\lambda$.
\end{lemma}


The strong convexity and smoothness of function are beneficial for model convergence, as the former can speed up the convergence rate while the latter can stabilize the training process. 
We first give the convergence results of the proposed methods (i.e., RFRec and RFRecF) without perturbation as follows.

\begin{theorem} (\textbf{RFRec})
\label{convegence_rate_RFRec}
When $0\le\alpha\le\frac{1}{L+\lambda}$, we have
\begin{equation}
   \mathbb{E}[\Vert x^k-x(\lambda)\Vert^2]\le(1-{\alpha\mu})^k\Vert x^0-x(\lambda)\Vert^2, \quad \forall k\in\mathbb{N}, 
\end{equation}
\end{theorem}

\begin{theorem} (\textbf{RFRecF})
\label{convegence_rate_RFRecF}
When $0<\alpha\le\frac{1}{2\mathcal{L}}$, we have
\begin{equation}
   \mathbb{E}[\Vert x^k-x(\lambda)\Vert^2]\le(1-{\alpha\mu})^k\Vert x^0-x(\lambda)\Vert^2+\frac{2\alpha\sigma^2}{\mu}, \quad \forall k\in\mathbb{N}, 
\end{equation}
where $\mathcal{L}:=
\max\{\frac{L}{1-p},\frac{\lambda}{p}\}$, and
$$\sigma^2:=\sum_{i=1}^{n}\left(\frac{1}{1-p}\Vert\nabla f_i(x_i(\lambda))\Vert^2+\frac{\lambda^2}{p}\Vert \boldsymbol{V}_{(i)}(\lambda)-\bar{\boldsymbol{V}}(\lambda)\Vert^2\right).$$
\end{theorem}
Generally, the variance term $\sigma^2$ will be eliminated with iterations. Specifically, When $k\rightarrow 0$, we have $\boldsymbol{V}_{(i)}\rightarrow\bar{\boldsymbol{V}}$ for all $i$, and $\nabla f(x(\lambda))\rightarrow 0$, therefore $\sigma^2\rightarrow 0$. So, the results reflect that both RFRec and RFRecF have \textbf{linear convergence rates}, and the iteration times are both $\frac{1}{\alpha\mu}\log\frac{1}{\varepsilon}$ to obtain $O(\varepsilon)$-accuracy. For RFRec, $\alpha\le\frac{1}{L+\lambda}$; for RFRecF, we have $\alpha\le\frac{1}{2\mathcal{L}}$ and $\mathcal{L}=
\max\{\frac{L}{1-p},\frac{\lambda}{p}\}\ge{L+\lambda}$ (equivalent when $p=\frac{\lambda}{L+\lambda}$). 
Then we conclude the results of optimal iteration times as follows.

\begin{corollary} 
\label{corollary:iter}
The optimal number of iterations of RFRec is $\frac{L+\lambda}{\mu}\log\frac{1}{\varepsilon}$ by selecting $\alpha = \frac{1}{L+\lambda}$. The optimal number of iterations of RFRecF is $\frac{2(L+\lambda)}{\mu}\log\frac{1}{\varepsilon}$ by selecting $p = \frac{\lambda}{L+\lambda}$ and $\alpha = \frac{1}{2(L+\lambda)}$.
\end{corollary}

RFRec obtains the same convergence rate as GD, $O(\log\frac{1}{\varepsilon})$. In contrast, the convergence rate of RFRecF seems worse than RFRec since each iteration only conducts one SGD step. Nevertheless, RFRecF can obtain a better communication efficiency than RFRec and will be shown in Section~\ref{sec:communication} and Section~\ref{sec:experiment_communication}.  

Then, we provide the convergence result of the proposed methods with perturbation as follows. 
\begin{corollary} (\textbf{Perturbation})
\label{trade-off}
The optimal iteration time of RFRec and RFRecF is $O(\log\frac{1}{\varepsilon})$ to obtain $(O(\varepsilon)+\frac{2\alpha s^2}{\mu})$-accuracy when adding the perturbation with noise strength $s$.
\end{corollary}

The above result indicates that the proposed methods still guarantee convergence. However, an additional error $\frac{2\alpha s^2}{\mu}$ emerges in the error neighbor, related to perturbation. Recall that the perturbation ensures $\frac{2\delta}{s}$-DP, which means a trade-off exists between model convergence and privacy. 


\subsection{Analysis of Communication Rounds}
\label{sec:communication}
In this section, we deduce the communication results of RFRec and RFRecF using the convergence results. Then, we discuss the parameter selection strategy to obtain optimal communication efficiency.

\subsubsection{\textbf{Optimal communication rounds}}
We first define the communication round~\cite{mcmahan2017communication} to represent communication cost. 
\begin{definition}
\label{def:communication}
A communication round is a round where clients send information (e.g., model parameters and gradients) to the server, or the server sends information to clients.  
\end{definition}
The communication rounds of RFRec are twice the iterations. Since at each iteration, the clients upload perturbed model $\widetilde{\boldsymbol{V}}_{(i)}$ to the server, and the server distributes average model $\bar{\boldsymbol{V}}$ to the clients. 

We consider the expected number of communication rounds of RFRecF because of the applied SGD manner. According to Algorithm~\ref{algorithm2}, communication is only needed at the current iteration when the current and previous iterations are at different places (e.g., client$\rightarrow$server, server$\rightarrow$client). Thus, the probability of a communication round at each iteration is $2p(1-p)$, and the expected number of communication rounds of RFRecF is as follows.


\begin{lemma}
    The expected numbers of communication rounds in $T$ iterations of RFRec and RFRecF are $2T$ and $2p(1-p)T$, respectively.
\end{lemma}
According to derivation of Corollary~\ref{corollary:iter}, we can obtain optimal $T$ of RFRecF by selecting $\alpha=\frac{1}{2\mathcal{L}}$. It is trivial to obtain the expected number of communication rounds of RFRecF is $\frac{4}{\mu}\max\{Lp,\lambda(1-p)\}\log\frac{1}{\varepsilon}$, where the optimal value of $p$ is $\frac{\lambda}{L+\lambda}$, the same as in Corollary~\ref{corollary:iter}. We conclude the communication results as follows. 

\begin{corollary} 
\label{corollary:commu}
    The optimal number of communication rounds of RFRec is $2\frac{L+\lambda}{\mu}\log\frac{1}{\varepsilon}$ by selecting $\alpha = \frac{1}{L+\lambda}$.
    The optimal expected number of communication rounds of RFRecF is $\frac{4\lambda}{L+\lambda}\frac{L}{\mu}\log\frac{1}{\varepsilon}$ by selecting $p = \frac{\lambda}{L+\lambda}$ and $\alpha = \frac{1}{2(L+\lambda)}$.
\end{corollary}

From Corollary~\ref{corollary:commu}, RFRec obtains high communication efficiency, $O(\log\frac{1}{\varepsilon})$, surpassing most FedRS methods with $O(\log m + \log\frac{1}{\varepsilon})$~\cite{jain2013low,lee2023randomly}. Surprisingly, RFRecF obtains a better communication efficiency than RFRec regardless of the value $\lambda$. Since $(L+\lambda)^2\ge4L\lambda$, RFRecF only needs at most half communication rounds of RFRec to obtain $O(\varepsilon)$-accuracy, theoretically. However, in reality, RFRecF cannot achieve such high efficiency, since $\sigma^2$ will not converge to $0$ rapidly, which may retard the convergence of the model.  

\subsubsection{\textbf{Parameter selection}}
\label{sec:para_selection}
We aim to select parameters for the proposed methods to achieve optimal communication efficiency. According to the theoretical results in Corollary~\ref{corollary:commu}, reducing the value of $\lambda$ can dwindle the communication rounds of both proposed methods. Nevertheless, $\lambda$ cannot be too small. Directly, Lemma~\ref{lemma:convex_2} and convergence will not hold if $\lambda$ is too small. When $\lambda$ decreases, the influence of penalty term $\psi(x)$ gets slight, and local model $\boldsymbol{V}_{(i)}$ moves slowly towards $\bar{\boldsymbol{V}}$, which obstacle the convergence of $\bar{\boldsymbol{V}}$ to the optimal global model.
Therefore, a moderate penalty parameter $\lambda$ is needed to balance the trade-off between recommendation performance and communication efficiency.      

If we know the lipschitz constant $L$, one intuitive way is selecting $\lambda=L$, then computing $\alpha=\frac{1}{2L}$ for RFRec and $\alpha=\frac{1}{4L}$ for RFRecF. This approach balances the importance of local and global loss (i.e., $f,\psi$), which is beneficial for training stability. However, $L$ is usually unknown in practice. In this case, we need to select $\lambda, \alpha$ by grid search. Fortunately, we can highly narrow the search range by using the relationship between the optimal choice of $\lambda$ and $\alpha$.
Supposing we select $\lambda=L$, the optimal choices are $\alpha=\frac{1}{2\lambda}$ for RFRec and $\alpha=\frac{1}{4\lambda}$ for RFRecF so that we can design the search ranges of $\lambda, \alpha$, correspondingly.
For example, the search range of $\lambda$ is $\{5,10,20,40\}$, while the range of $\alpha$ is $\{0.1,0.05,0.025,0.0125\}$. 
In particular, for RFRecF, we let $p=0.5$ ($p = \frac{\lambda}{L+\lambda}$), which means the equal weights of the main update and moving to average.

\begin{table}
\vspace{-2mm}
\fontsize{8}{11}\selectfont
  \caption{Statistics of the datasets.}
  \vspace{-2mm} 
  \label{tab:statistics}
  \begin{tabular}{cccc}
    \toprule
    Datasets & \# Users & \# Items & \# Interactions \\
    \midrule
    \textbf{ML-100k} & $943$ & $1,682$ & $100,000$ \\
    \textbf{ML-1M} & $6,041$ & $3,706$ & $1,000,209$ \\
    \textbf{KuaiRec} & $7,176$ & $10,728$ & $12,530,806$  \\
    \textbf{Jester} & $73,421$ & $100$ & $4,136,360$  \\
    \bottomrule
  \end{tabular}
  \vspace{-4mm} 
  \label{tab:Table_1}
\end{table}

\begin{table*}[t]
\vspace{-2mm} 
    \fontsize{8}{11}\selectfont
    \caption{Overall recommendation performance comparison. All improvements are \textbf{statistically significant} (i.e., two-sided t-test with $p<0.05$) over federated baselines. In each row, the best result of federated methods is bold.}
        \vspace{-2mm} 
    \resizebox{\linewidth}{!}{\begin{threeparttable}
    \begin{tabular}{ccccccccccc}
        \toprule
        \toprule
        \multirow{2}{*}{{Datasets}}&\multirow{2}{*}{{Metrics}}&
        \multicolumn{2}{c}{{Centralized}} & \multicolumn{5}{c}{{Federated}} & \multicolumn{2}{c}{{Ours}}\cr
        \cmidrule(lr){3-11}
        & & PMF & SVD++ & FCF & FedRec & FedFast & SemiDFEGL & FedNCF & \textbf{RFRec} & \textbf{RFRecF}\cr
        \cmidrule(lr){1-11}
        \multirow{2}{*}{ML-100k}
        & MAE  & $0.7341\pm0.0014$ & $0.7289\pm0.0011$
        & $0.7602\pm0.0026$ & $0.7549\pm0.0041$ & $0.7545\pm0.0040$ & $0.7583\pm0.0019$ & $0.7562\pm0.0011$
        & $\bm{0.7237\pm0.0015}$ & $0.7317\pm0.0064$ \cr
        & RMSE & $0.9318\pm0.0026$ & $0.9273\pm0.0025$
        & $0.9614\pm0.0028$ & $0.9832\pm0.0046$ & $0.9403\pm0.0051$ & $0.9645\pm0.0025$ & $0.9612\pm0.0024$
        & $\bm{0.9325\pm0.0023}$ & $0.9385\pm0.0077$ \cr
        \cmidrule(lr){1-11}
        \multirow{2}{*}{ML-1m}
        & MAE  & $0.6878\pm0.0029$ & $0.6924\pm0.0009$
        &$0.7014\pm0.0018$ & $0.7289\pm0.0023$ & $0.7056\pm0.0085$ & $0.7130\pm0.0038$ & $0.7090\pm0.0047$
        & $0.6937\pm0.0034$ & $\bm{0.6906\pm0.0033}$ \cr
        & RMSE & $0.8762\pm0.0031$ & $0.8863\pm0.0015$
        & $0.8874\pm0.0022$ & $0.9161\pm0.0025$ & $0.8860\pm0.0095$ & $0.8937\pm0.0049$ & $0.8966\pm0.0050$
        & $\bm{0.8831\pm0.0038}$ & $0.8840\pm0.0042$ \cr
        \cmidrule(lr){1-11}
        \multirow{2}{*}{KuaiRec}
        & MAE  & $0.3814\pm0.0012$ & $0.3760\pm0.0010$
        & $0.4861\pm0.0031$ & $0.5949\pm0.0015$ & $0.4678\pm0.0032$ & $0.4798\pm0.0027$ & $0.3934\pm0.0052$
        & $\bm{0.3684\pm0.0017}$ & $0.3781\pm0.0022$ \cr
        & RMSE & $0.7275\pm0.0015$ & $0.7145\pm0.0015$
        & $0.7569\pm0.0029$ & $0.7969\pm0.0030$ & $0.7387\pm0.0038$ & $0.7511\pm0.0031$ & $0.7474\pm0.0056$
        & $\bm{0.7163\pm0.0025}$ & $0.7280\pm0.0049$ \cr
        \cmidrule(lr){1-11}
        \multirow{2}{*}{Jester}
        & MAE  & $0.8062\pm0.0012$ & $0.8006\pm0.0008$
        & $0.8175\pm0.0021$ & $0.8473\pm0.0011$ & $0.8321\pm0.0066$ & $0.8275\pm0.0029$ & $0.8330\pm0.0030$
        & $0.8118\pm0.0017$ & $\bm{0.8086\pm0.0008}$ \cr
        & RMSE & $1.0440\pm0.0032$ & $1.0430\pm0.0018$
        & $1.0718\pm0.0064$ & $1.0755\pm0.0012$ & $1.0982\pm0.0127$ & $1.0726\pm0.0038$ & $1.0712\pm0.0028$
        & $\bm{1.0447\pm0.0021}$ & $1.0565\pm0.0035$ \cr
        \bottomrule
        \bottomrule
    \end{tabular}
    \vspace{0cm}
    \end{threeparttable}}
        \label{tab:recommendation}
        \vspace{-4mm}
\end{table*}

\section{Experiments}
\label{sec:experiment}
In this section, we aim to answer the following research questions:
\begin{itemize}[leftmargin=*]
\item \textbf{RQ1}: How does the recommendation performance of proposed RFRec and RFRecF compared to state-of-the-art FedRS models?
\item \textbf{RQ2}: Does RFRec and RFRecF obtain better communication efficiency compared to FedRS models?
\item \textbf{RQ3}: How's the contribution of each component?
\item \textbf{RQ4}: How robust are the proposed methods?
\item \textbf{RQ5}: How RFRec and RFRecF influenced by parameter?
\item \textbf{RQ6}: How to balance the privacy protection and performance?
\end{itemize}

\subsection{Experimental Settings}

\subsubsection{\textbf{Datasets and Evaluation Metrics}} 
To evaluate the effectiveness of the proposed RFRec and RFRecF, we conduct experiments on four benchmark datasets: (1) \textbf{ML-100k}: It contains $100$ thousands of user interactions (i.e., ratings) on movies. (2) \textbf{ML-1m}: It contains $1$ million user interactions on movies. (3) \textbf{KuaiRec}: It contains $12,530,806$ user interactions (i.e., watch ratio) on video. (4) \textbf{Jester}: It contains about $4,136,360$ user ratings on jokes. To prevent the influence of extreme values on the experiment, we filter out the data with a watch ratio greater than $20$ in the KuaiRec dataset, losing less than $1$\textperthousand~data. The statistics of the datasets are shown in Table~\ref{tab:statistics}.
We consider \textit{MAE} and \textit{RMSE} as the evaluation metrics, and each is widely used in the recommendations.

\subsubsection{\textbf{Baselines}}
To demonstrate the effectiveness of our model, we compare RFRec with two centralized RS methods (PMF, SVD++) and state-of-the-art FedRS methods. 

\begin{itemize}[leftmargin=*]
\item \textbf{PMF}~\cite{mnih2007probabilistic}: One of the most popular MF-based recommendation methods, modeling the latent features of users and items. 
\item \textbf{SVD++}~\cite{koren2008factorization}: A popular variant of SVD considers user and item biases to enhance recommendation performance.
\item \textbf{FCF}~\cite{ammad2019federated}: The first MF-based federated recommendation method.
\item \textbf{FedRec}~\cite{lin2020fedrec}: A federated recommendation method that generates virtual ratings for privacy-preserving.
\item \textbf{FedFast}~\cite{muhammad2020fedfast}: A federated method using clustering and sampling to achieve high communication efficiency. 
\item \textbf{SemiDFEGL}~\cite{qu2023semi}: A method leverages federated ego graph.
\item \textbf{FedNCF}~\cite{perifanis2022federated}: An NCF-based method integrating FedAvg~\cite{mcmahan2017communication}.
\end{itemize}

\subsubsection{\textbf{Implementation Details}}
Following previous studies~\cite{koren2008factorization,mnih2007probabilistic,wu2021fedgnn}, we use Gaussian distribution to initialize trainable parameters. 
We optimize baselines utilizing Adam~\cite{kingma2014adam}.   
From the suggestions of previous works~\cite{mnih2007probabilistic,ammad2019federated,lin2020fedrec}, we set the size of latent features as $d=20$ and the maximum iteration number as $K=100$. For all baseline models, we set penalty parameters as $\lambda_u = \lambda_v = 0.1$. For our proposed model, we select the learning rate $\alpha=0.05$ and $\alpha=0.025$ for RFRec and RFRecF, respectively. In addition, we set the penalty parameter $\lambda = 10$ and threshold $p=0.5$. Moreover, we provide parameter analysis in Section~\ref{sec:para_analysis}.
We implement our proposed methods\footnote{\url{https://github.com/Applied-Machine-Learning-Lab/RFRec}} in Python 3.10.11 and Pytorch 2.0.1+cu118.


\subsection{Overall Performance Comparison (\textbf{RQ1})}
To demonstrate the effectiveness of our proposed methods, RFRec and RFRecF, for improving recommendation performance in the FL setting, we compare our model with centralized baselines (PMF, SVD++) and federated baselines (FCF, FedRec, FedFast, SemiDFEGL, FedNCF) in four datasets. We report the main results of recommendation performance in Table~\ref{tab:recommendation}, where we can observe that:
\begin{itemize}[leftmargin=*]
    \item In general, RFRec and RFRecF perform significantly better than other federated baselines on all datasets, demonstrating the superiority of the proposed methods. We attribute the improvements to that the proposed new formulation~\eqref{RFL_opt} of FedRS can well balance the learning of local and global model. In addition, the proposed methods, RFRec and RFRecF, both provide stable training process, which guarantee convergence to the optimal model. 
    \item Significantly, the recommendation performances of the proposed methods are highly close to centralized methods, demonstrating the effectiveness of our methods. We attribute the excellent performance to the fact that our proposed optimization problem~\eqref{RFL_opt} can well approximate the original recommendation problem~\eqref{MF_opt} by selecting suitable penalty parameter $\lambda$. 
    \item Comparing the recommendation performance of our methods, RFRec outperforms RFRecF in most cases, as the GD manner is more stable than the SGD manner, which reflects a trade-off between recommendation accuracy and efficiency (SGD converges faster). 
\end{itemize}

\subsection{Communication Rounds (\textbf{RQ2})}
\label{sec:experiment_communication}

\begin{figure}[t]
    \centering
    \subfigure[Speed on ML-1m.]{
        \begin{minipage}[t]{0.475\linewidth}
            \includegraphics[width=1.05\linewidth]{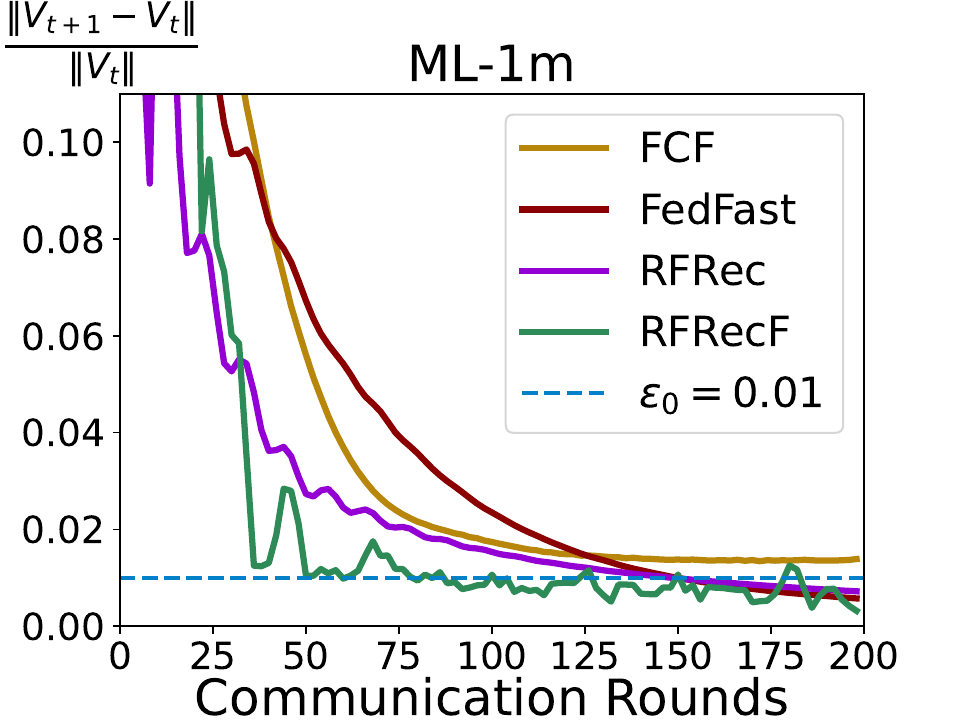}
            \vspace{-6mm}
        \label{subfig: ML-convergence}
        \end{minipage}
    }
    \subfigure[Communication rounds on ML-1m.]{
        \begin{minipage}[t]{0.475\linewidth}
            \includegraphics[width=1.05\linewidth]{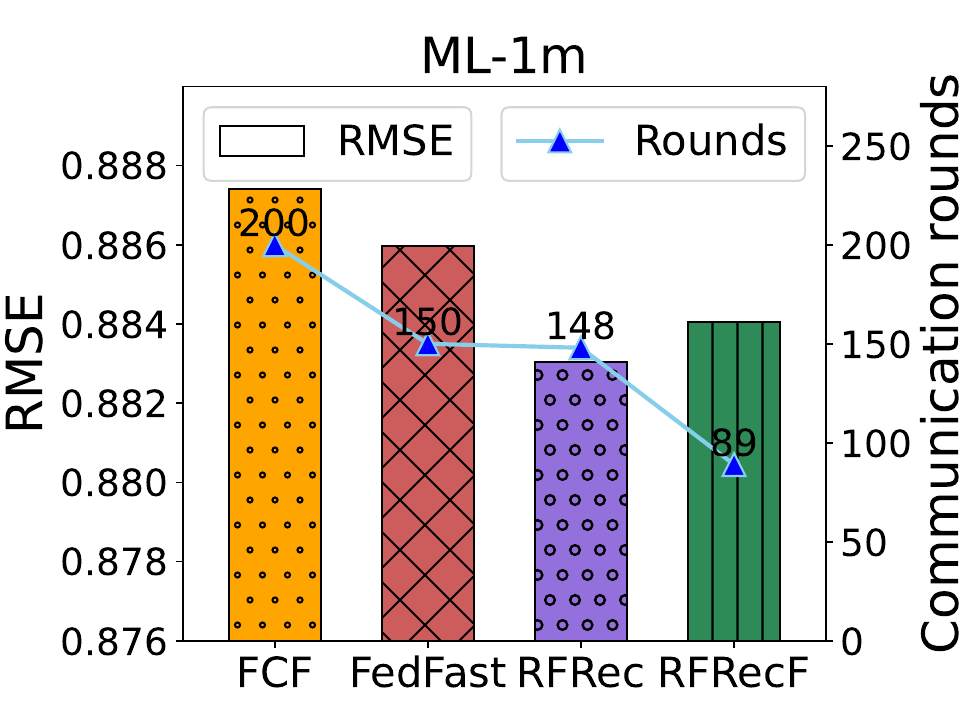}
            \vspace{-6mm}
        \label{subfig: ML-communication}
        \end{minipage}
    }
    \subfigure[Speed on Jester.]{
        \begin{minipage}[t]{0.475\linewidth}
            \includegraphics[width=1.05\linewidth]{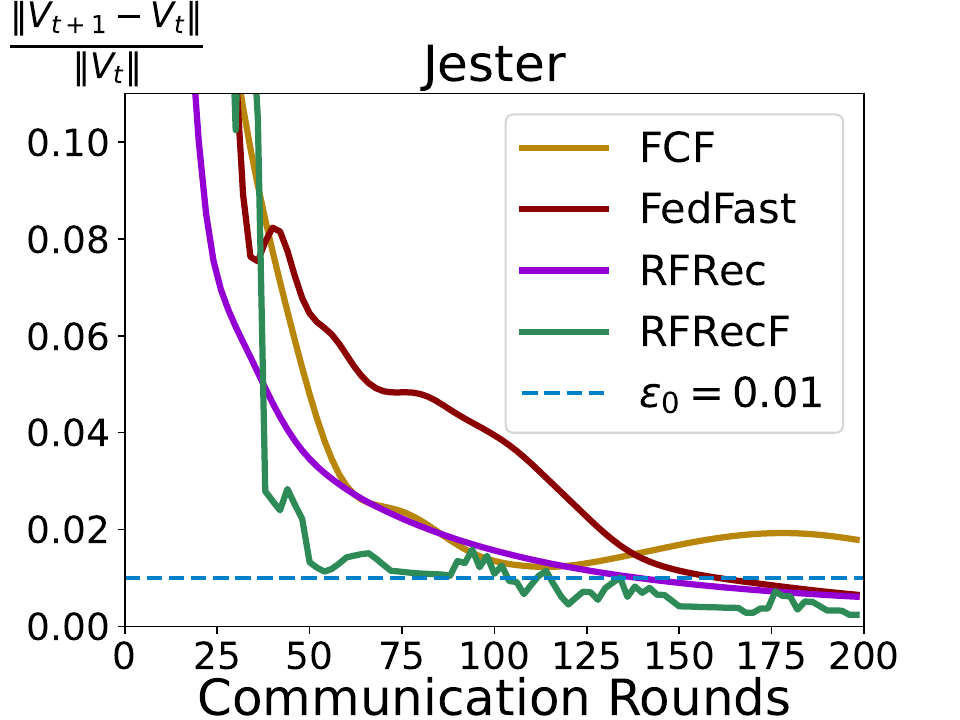}
            \vspace{-6mm}
        \label{subfig: Jester-convergence}
        \end{minipage}
    }
    \subfigure[Communication rounds on Jester.]{
        \begin{minipage}[t]{0.475\linewidth}
            \includegraphics[width=1.05\linewidth]{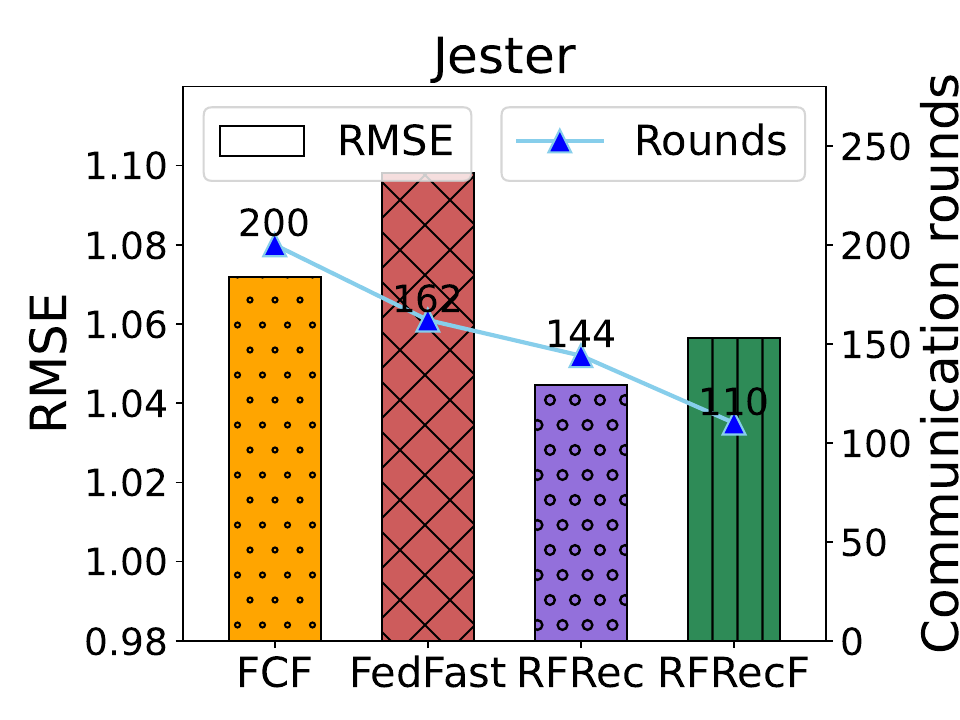}
            \vspace{-6mm}
        \label{subfig: Jester-communication}
        \end{minipage}
    }
    \vspace{-6mm}
    \caption{Communication results.}
    \label{fig:communication}
    \vspace{-6mm}
\end{figure}

In this section, we evaluate the communication efficiency of the proposed methods. 
We use FCF and FedFast as baselines to compare the communication costs, as they have a relatively good recommendation performance among all baselines. As the communication medium for the baselines and our methods have the same size (i.e., gradient and item matrix are both $\mathbb{R}^{d\times m}$), we utilize the communication rounds as the metric. The communication will terminate when the model converges. The target of FedRS is to train an optimal item matrix $\boldsymbol{V}$, so we set the stop criterion of communication as $\frac{\Vert \boldsymbol{V}_{t+1}-\boldsymbol{V}_t\Vert}{\Vert \boldsymbol{V}_t\Vert}\le\varepsilon_0$ and the maximum communication number as $200$. We report the communication results in Figure~\ref{fig:communication}. 
\begin{itemize}[leftmargin=*]
    \item The communication efficiency of RFRec and RFRecF consistently surpasses federated baselines on two datasets, demonstrating the effectiveness of the proposed method. Specifically, the communication efficiency of FedFast outperforms FCF since it uses the clustering and sampling approach to speed up the training process. Although RFRec only slightly exceeds FedFast in communication efficiency, it provides much better recommendation performance. In particular, RFRecF significantly leads among all methods regarding communication efficiency. Because RFRec leverages a flexible SGD manner to reduce communication rounds. 
    \item In fact, the number of communication rounds of FCF is larger than the maximum number ($200$), as it does not achieve the stop criterion. Since the optimal number of communication rounds is $O(\log m + \log\frac{1}{\varepsilon})$, which is gigantic when item number $m$ is large. 
    \item Again, the results reflect the trade-off between efficiency and accuracy. While RFRec has a more stable convergence process and better recommendation performance, RFRecF converges faster, providing higher communication efficiency.   
\end{itemize}

\subsection{Ablation Study (\textbf{RQ3})}
\label{sec:ablation}

\begin{figure}[t]
    \centering
    \includegraphics[width=1.05\linewidth]{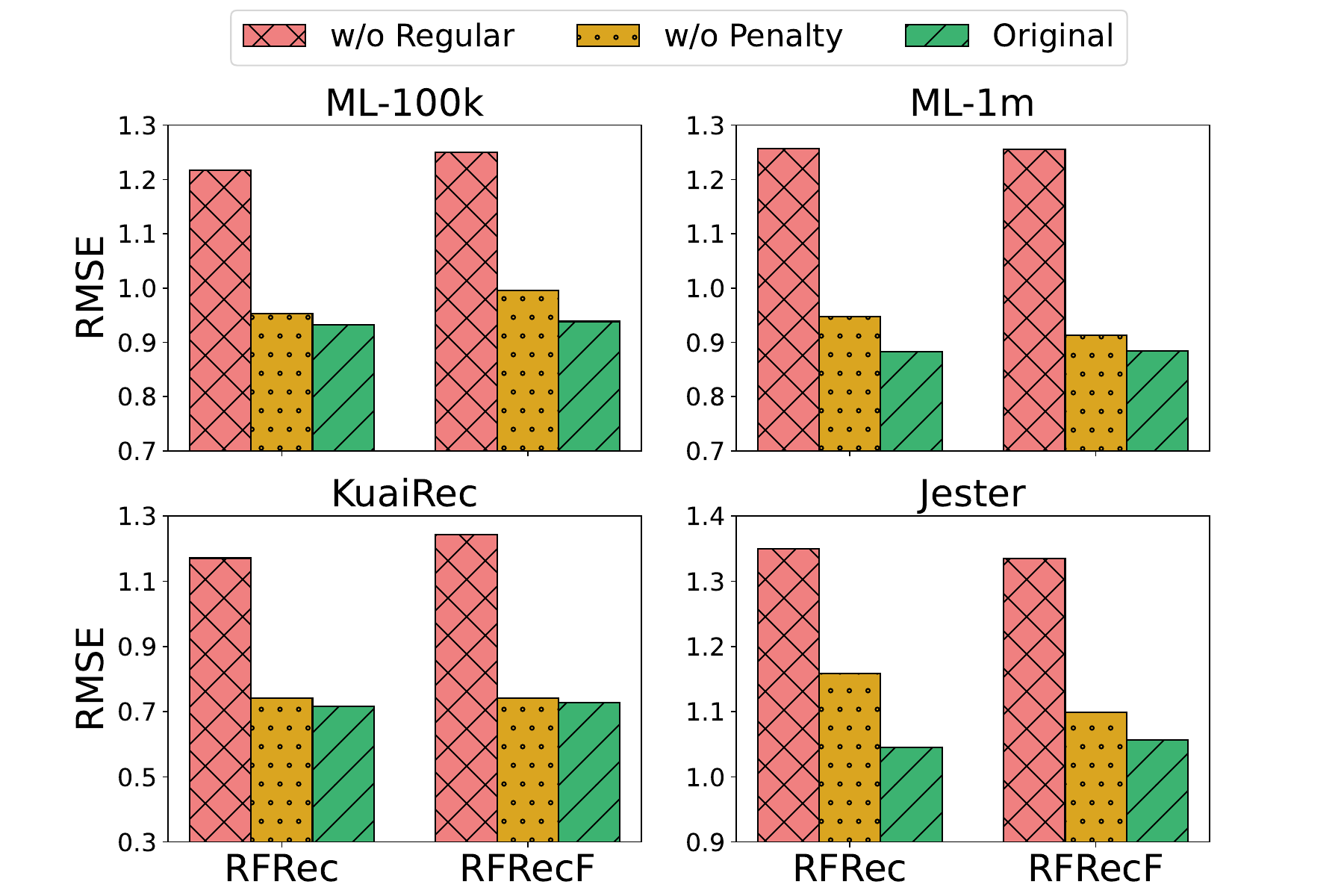}
    \vspace{-8mm}
    \caption{Impact of different components.}
    \label{fig:ablation}
    \vspace{-4mm}
\end{figure}

To test the contributions of two critical parts, we conduct the ablation study at two variants of RFRec and RFRecF over four datasets, including (1) \textit{w/o Regular}: without the regularization term $\psi$, and (2) \textit{w/o Penalty}: without the penalty term $\Vert \boldsymbol{U}\Vert^2$. We record the RMSE to compare the recommendation performances, as demonstrated in Figure~\ref{fig:ablation}. The regularization term is the most critical part, which contributes the most to performance since it guarantees the convergence of each local model $\boldsymbol{V}_{(i)}$ to global model $\bar{\boldsymbol{V}}$. Furthermore, the penalty term can improve performance in all cases by constraining the scale of model parameters to alleviate overfitting.


\subsection{Robustness Test (\textbf{RQ4})}

\begin{table}[t]
    \fontsize{8}{11}\selectfont
    \caption{Robustness of recommendation performance.}
        \vspace{-2mm} 
    \begin{threeparttable}
    \begin{tabular}{ccccccc}
        \toprule
        \toprule
        \multirow{1}{*}{Methods}&\multirow{1}{*}{Metrics}&
        $0\%$ & $20\%$ & $50\%$ & $80\%$ &
        $90\%$ \cr
        \cmidrule(lr){1-7}
        \multirow{2}{*}{RFRec}
        & MAE  & $0.6937$ & $0.6958$ & $0.7083$ & $0.7090$ & $0.7117$ \cr
        & RMSE & $0.8831$ & $0.8870$ & $0.8956$ & $0.8966$ & $0.9001$ \cr
        \cmidrule(lr){1-7}
        \multirow{2}{*}{RFRecF}
        & MAE  & $0.6906$ & $0.6974$ & $0.7017$ & $0.7038$ & $0.7064$ \cr
        & RMSE & $0.8840$ & $0.8891$ & $0.9005$ & $0.9048$ & $0.9057$ \cr
        \bottomrule
        \bottomrule
    \end{tabular}
    \end{threeparttable}
        \label{tab:robust_analysis}
        \vspace{-2mm}
\end{table}

We experiment on ML-1m to demonstrate the robustness of our methods, with a rate of the dropped devices from $0\%$ to $90\%$ to simulate the lousy case (e.g., losing connection). We report the results in Table~\ref{tab:robust_analysis}. Our methods are robust enough to provide a good recommendation performance even in extreme cases where the participation rate is meager, from $10\%$ to $50\%$. The results demonstrate the effectiveness of our proposed methods in real scenarios, such as industrial scenarios, where slight performance degradation due to losing connections of devices is acceptable.

\subsection{Hyper-Parameter Analysis (\textbf{RQ5})}
\label{sec:para_analysis}

\begin{table}[t]
    \fontsize{8}{11}\selectfont
    \caption{Parameter sensitivity of RMSE on ML-1m.}
        \vspace{-4mm} 
    \begin{threeparttable}
    \begin{tabular}{cccccc}
        \toprule
        \toprule
        \multirow{1}{*}{Methods}&\multirow{1}{*}{Parameter}&
        $\lambda=5$ & $\lambda=10$ & $\lambda=20$ & $\lambda=40$\cr
        \cmidrule(lr){1-6}
        \multirow{4}{*}{RFRec}
        & $\alpha=0.1000$  & $0.8912$ & $0.8948$ & $0.8963$ & $0.9036$ \cr
        & $\alpha=0.0500$  & $0.8957$ & $\textbf{0.8846}$ & $0.8869$ & $0.8887$\cr
        & $\alpha=0.0250$  & $0.9057$ & $0.8928$ & $0.8897$ & $0.8983$ \cr
        & $\alpha=0.0125$  & $0.8895$ & $0.8870$ & $0.8891$ & $0.8920$ \cr
        \cmidrule(lr){1-6}
        \multirow{4}{*}{RFRecF}
        & $\alpha=0.1000$  & $0.9169$ & $0.9035$ & $0.9066$ & $0.9139$ \cr
        & $\alpha=0.0500$  & $0.9254$ & $0.9108$ & $0.8923$ & $0.9070$ \cr
        & $\alpha=0.0250$  & $0.9101$ & $\textbf{0.8830}$ & $0.8856$ & $0.8908$ \cr
        & $\alpha=0.0125$  & $1.0491$ & $0.9588$ & $0.9102$ & $0.8992$ \cr
        \bottomrule
        \bottomrule
    \end{tabular}
    \vspace{-2mm}
    \end{threeparttable}
        \label{tab:para_analysis}
        \vspace{-2mm}
\end{table}
In this section, we tune the important hyperparameters $\lambda,\alpha,p$ to find the best choice. Specifically, we conduct search on $\lambda$ and $\alpha$ in range $\{0.1,0.05,0.025,0.0125\}$ and $\{5,10,20,40\}$, according to Section~\ref{sec:para_selection}. The experiment results are reported in Table~\ref{tab:para_analysis}, where a general observation is that the performance in the margin of the search range is worse than the center. In particular, the optimal choice of $\lambda$ is $10$ for two methods. While the optimal choices of $\alpha$ are $0.05$ and $0.025$ for RFRec and RFRecF, respectively. The optimal choice conforms to the theoretical deduction in Section~\ref{sec:para_selection}. Then, we tune value $p$ for RFRecF after obtaining the best $\lambda,\alpha$. We find the overall trend of RMSE is a U curve, with optimal value $p=0.5$. 


\subsection{Privacy Protection Test (\textbf{RQ6})}
\label{sec:Privacy}
To show the trade-off between privacy-preserving and performance mentioned at Corollary~\ref{trade-off} and determine the best parameters in Equation~\eqref{eq:perturbation}, we experiment on ML-1m dataset by tuning clipping threshold $\delta$ and noise scale $s$. The privacy budget is computed by $\epsilon = \frac{2\delta}{s}$, the smaller the better. We report the main result in Figure~\ref{fig:privacy}.
Generally, there exists a trade-off between privacy protection and accuracy. For example, when we amplify $s$, RMSE gets larger, which means worse recommendation performance, and the privacy budget $\epsilon$ becomes smaller, meaning better privacy strength. In contrast, $\delta$ has a contrary effect.  
We choose parameter pair $\delta=0.2,s=0.04$ for RFRec, and parameter pair $\delta=0.2,s=0.06$ for RFRec, providing enough privacy protection and nice performance.


\begin{figure}[t]
    \centering
    \subfigure[Privacy budget and RMSE of RFRec.]{
        \begin{minipage}[t]{0.475\linewidth}
            \includegraphics[width=1.05\linewidth]{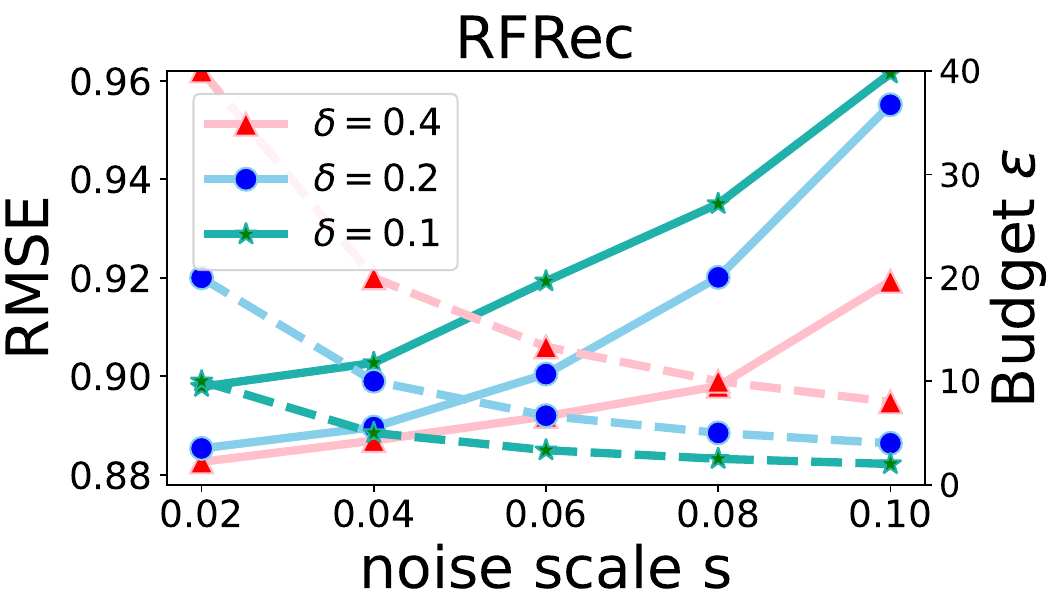}
        \label{subfig: privacy_RFRec}
        \vspace{-4mm}
        \end{minipage}
    }
    \subfigure[Privacy budget and RMSE of RFRecF.]{
        \begin{minipage}[t]{0.475\linewidth}
            \includegraphics[width=1.05\linewidth]{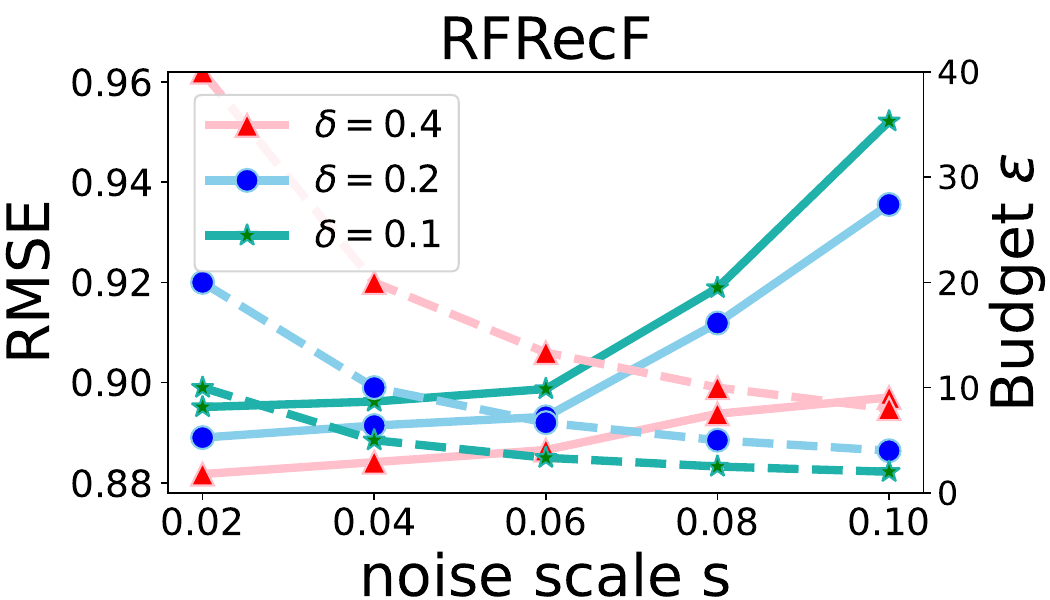}
        \label{subfig: privacy_RFRecF}
        \vspace{-4mm}
        \end{minipage}
    }
    \vspace{-4mm}
    \caption{Privacy budget and RMSE.}
    \label{fig:privacy}
   \vspace{-4mm}
\end{figure}

\section{Related works}
The mainstream of FedRs~\cite{ammad2019federated,lin2020fedrec,muhammad2020fedfast,liang2021fedrec++} is based on Matrix factorization (MF), which assumes that the rating is the inner product of user and item feature vectors. 
To preserve user privacy, FedRS models divide the update into server and client. 
However, the communication of gradients cannot protect user information strictly. In FedMF~\cite{chai2020secure}, it has been proved that the central server can deduce a user's ratings by leveraging two consecutive gradients of this user. Motivated by this severe hidden threat to user privacy, some recent works~\cite{liu2023privaterec,chai2020secure,minto2021stronger,liu2022federated} make efforts to enhance privacy protection through homomorphic encryption, pseudo-item generation, and LDP. Other methods~\cite{wu2021fedgnn,wu2022federated,qu2023semi,lin2020meta} try to improve the capability of FedRS in some special domains. For example, FedGNN~\cite{wu2021fedgnn} integrated the graph neural network (GNN) into the framework of FL to improve the representation learning and simultaneously communicate the perturbed gradient to alleviate privacy leaks. Besides, to equip itself with the ability to handle users' heterogeneous (Non-IID) data, MetaMF~\cite{lin2020meta} leverages the meta-learning technique to allow personalization in local clients. Detecting the incompatibility of federated aggregation methods and NCF model update, FedNCF~\cite{perifanis2022federated} propose a framework that decomposes the aggregation in the MF step and model averaging step. 

The mainstream of FedRS methods cannot guarantee stable convergence since solving the non-convex problem and leveraging the alternating update approach.
Our methods guarantee the convergence to the optimal and provide high communication efficiency.  

\section{Conclusion}
In this paper, we reformulate the optimization problem of FedRS as an RERM problem to adapt to the FL setting. To effectively and efficiently solve the proposed optimization problem, we offer two methods, RFRec and RFRecF, leveraging local GD and non-uniform SGD manners, respectively, to learn the optimal model parameters.
We provide the theoretical analysis of convergence and communication, showcasing our proposed methods' stabilization and high communication efficiency. In addition, we demonstrate the privacy protection mechanism of the proposed methods. 
Extensive experiments demonstrate the superiority and effectiveness of our proposed methods over other state-of-the-art FedRS methods. 


\appendix
\section{Appendix}
\label{appendix}

\begin{proof} (Lemma~\ref{lemma:convex_1}) 
It is trivial to proof $\Vert \nabla f_i(x) - \nabla f_i(y)\Vert\le L\Vert x-y\Vert$.
Then, we can use the property that $L$-Lipschitz continuity of gradients is equivalent to $L$-smoothness to prove the smoothness.

\end{proof}

\begin{proof} (Lemma~\ref{lemma:convex_2})
Recall that 
$F_i(x) = \Vert{R_i} - u^T V\Vert^2 + \lambda_u\Vert u\Vert^2 + \frac{\lambda}{2}\Vert V_{(i)}-\bar{V}\Vert^2,$ we can derive the first order derivative as 
\begin{equation}
    \nabla F_i = \begin{bmatrix} -2V(R_i-u^T V)^T \\ -2u(R_i - u^T V)\end{bmatrix} + 
    \begin{bmatrix} 2\lambda_u u \\ (1-\frac{1}{n})\lambda V \end{bmatrix},
\end{equation}
Then we have
\begin{equation}
    \nabla^2 F_i = 2\begin{bmatrix} V V^T + \lambda_u I_d & 2V\otimes u^T-R_i\otimes I_d \\ 2V^T\otimes u - R_i^T\otimes I_d & I_m\otimes u u^T+\frac{1}{2}(1-\frac{1}{n})\lambda I_{md}\end{bmatrix}.
\end{equation}
The Hessian matrix can be decomposed as follows:
\begin{equation}
\begin{aligned}
    \nabla^2 F_i &= 2\begin{bmatrix} V V^T & 2V\otimes u^T \\ 2V^T\otimes u & I_m\otimes u u^T\end{bmatrix} + 2\begin{bmatrix} \lambda_u I_d & -R_i\otimes I_d \\ -R_i^T\otimes I_d & \frac{1}{2}(1-\frac{1}{n})\lambda I_{md}\end{bmatrix} \\
    &= 2A^T A + \frac{2}{\lambda_u} B^T B + \frac{2}{\lambda_u}\begin{bmatrix} 0 & 0 \\ 0 & \frac{1}{2}(1-\frac{1}{n})\lambda_u\lambda I_{md} - C\end{bmatrix}, 
\end{aligned}
\end{equation}
where $A=\begin{bmatrix} V^T & 2I_m\otimes u^T \end{bmatrix}, B=\begin{bmatrix} \lambda_u I_d & -R_i\otimes I_d \end{bmatrix}, C = R^T_i R_i\otimes I_d + 3\lambda_u I_m\otimes u u^T .$
Notice that $C$ is the sum of two symmetric matrices. According to Weyl's inequality and the eigenvalues of Kronecker product are the pair-wise products of eigenvalues of matrices, we have $\lambda_{\max}(C)\le\lambda_{\max}(R^T_i R_i)+3\lambda_u\lambda_{\max}(u u^T)$.
Since $\text{rank}(R^T_i R_i)=1$ and $\text{tr}(R^T_i R_i)=\text{tr}(R_i R^T_i)=\Vert R_i\Vert^2$, we have $\lambda_{\max}(R^T_i R_i)=\Vert R_i\Vert^2\le M_r^2$.
Similarly, we have $\lambda_{\max}(u u^T)=\Vert u\Vert^2\le M_u^2$.
By letting $\lambda>\frac{2}{\lambda_u}M_r^2+6M_u^2$, we have $F_i$ is $\mu$-strongly convex ($M^T M$ is positive definite). Thus, F is $\mu$-strongly convex.

Then, we can leverage the smoothness of $f_i$ and analyze the Hessian of $\psi$ to prove the smoothness. 
Observe that
$$ \nabla^2 f(x)=\text{diag}(\nabla^2 f_1(x_1),\nabla^2 f_2(x_2),\cdots,\nabla^2 f_n(x_n)). $$
Then by Lemma~\ref{lemma:convex_1}, we have $\nabla^2 f_i(x_i)\preceq L I_{(m+1)d}$, which implies $\nabla^2 f(x)\preceq L I_{(m+1)nd}$
and thereby $f$ is $L$-smooth. 

For function $\psi$, we have $\nabla_V\psi(x)=(V_{(1)}-\bar{V}, \cdots,V_{(n)}-\bar{V})^\mathrm {T}.$
By denoting $e$ as the all one vector in $\mathbb{R}^{n}$, we have 
$$\nabla_V^2\psi(x)=(I_n-\frac{1}{n}ee^T)\otimes I_{md},$$
where $I_n-\frac{1}{n}ee^T$ is a circulant matrix with maximum eigenvalue $1$. 
Hence $\nabla^2\psi(x)\preceq I_{(m+1)nd}$, and $\psi$ is $1$-smooth. It follows that $F$ is $\mu$-strongly convex and $L_F$-smooth with $L_F = L+\lambda$.
\end{proof}

\begin{proof} (Theorem~\ref{convegence_rate_RFRec})
Denote $x^*=x(\lambda)$ and the residual $r^k=x^k-x^*$. 
By the update rule in Algorithm~\ref{algorithm1}, we have $r^{k+1}=r^k-\alpha \nabla F(x^k)$. Taking inner products for both sides, we obtain 
\begin{align*}
\Vert r^{k+1}\Vert^2 &=\Vert r^k\Vert^2 -2\alpha \langle r^k,\nabla F(x^k) \rangle + \alpha^2\Vert \nabla F(x^k)\Vert^2\\
&\le (1-\alpha\mu)\Vert r^k\Vert^2 -2\alpha(F(x^k)-F(x^*)) + \alpha^2\Vert \nabla F(x^k)\Vert^2\\
&\le (1-\alpha\mu)\Vert r^k\Vert^2 + 2\alpha(\alpha L-1)(F(x^k)-F(x^*)) \\
&\le (1-\alpha\mu)\Vert r^k\Vert^2\quad(0\le\alpha\le\frac{1}{L+\lambda}),
\end{align*}
where we use $\mu$-strongly convexity and $L$-smoothness of $F$. 
\end{proof}

Before the main proof of Theorem~\ref{convegence_rate_RFRecF}, we give the result of expected smoothness of $G$ as follows. 
\begin{lemma} (Expected Smoothness)
\label{lemma_3}
For every $x\in \mathbb{R}^d$, the stochastic gradient $G$ of $F$ satisfies 
\[\mathbb{E}[\Vert G(x)-G(x(\lambda))\Vert^2]\le 2\mathcal{L}(F(x)-F(x(\lambda))),\]
\[\mathbb{E}[\Vert G(x)\Vert^2]\le 4\mathcal{L}(F(x)-F(x(\lambda)))+2\sigma^2.\]
\end{lemma}

\begin{proof}
By the definition of the stochastic gradient $G$, we have 
\begin{align*}
&\mathbb{E}[\Vert G(x)-G(x(\lambda))\Vert^2]\leq \frac{2L_f}{1-p} D_f(x,x(\lambda))+\frac{2\lambda^2 L_{\psi}}{p}D_\psi(x,x(\lambda))
\end{align*}

where $D_f, D_{\psi}$ are bregman distances. Since $D_f+\lambda D_\psi = D_F$ and $\nabla F(x(\lambda))=0$, we can continue as 
\begin{align*}
\mathbb{E}[\Vert G(x)-G(x(\lambda))\Vert^2]&\le
2\max\{\frac{L}{1-p},\frac{\lambda}{p}\}D_F(x,x(\lambda))\\
&=2\mathcal{L}(F(x)-F(x(\lambda))).
\end{align*}
This proves the first bound. For the second estimate, we have
\begin{align*}
    \mathbb{E}[\Vert G(x)\Vert^2]&\le 2\mathbb{E}[\Vert G(x)-G(x(\lambda))\Vert^2]+2\mathbb{E}[\Vert G(x(\lambda))\Vert^2]\\
    &\le 4\mathcal{L}(F(x)-F(x(\lambda)))+2\sigma^2,
\end{align*}
where it is trivial to show $\sigma^2 = \mathbb{E}[\Vert G(x(\lambda))\Vert^2]$.
\end{proof}

\begin{proof} (Theorem~\ref{convegence_rate_RFRecF})
According to Algorithm~\ref{algorithm2}, we have $r^{k+1}=r^k-\alpha G(x^k)$. Then similar to the proof of RFRec we have

\[\Vert r^{k+1}\Vert^2=\Vert r^k\Vert^2 -2\alpha \langle r^k,G(x^k) \rangle + \alpha^2\Vert G(x^k)\Vert^2\]
By taking expectation on the both sides, we obtain

\begin{align*}
\mathbb{E}[\Vert r^{k+1}\Vert^2] &=\Vert r^k\Vert^2 -2\alpha \langle r^k,\nabla F(x^k) \rangle + \alpha^2\mathbb{E}[\Vert G(x^k)\Vert^2]\\
&\le (1-\alpha\mu)\Vert r^k\Vert^2 -2\alpha(F(x^k)-F(x^*)) + \alpha^2\mathbb{E}[\Vert G(x^k)\Vert^2]\\
&\le (1-\alpha\mu)\Vert r^k\Vert^2 + 2\alpha(2\alpha\mathcal{L}-1)(F(x^k)-F(x^*))+2\alpha^2\sigma^2 \\
&\le (1-\alpha\mu)\Vert r^k\Vert^2 + 2\alpha^2\sigma^2 \quad(0\le\alpha\le\frac{1}{2\mathcal{L}}),
\end{align*}
where the second inequality applies the result in Lemma~\ref{lemma_3}.
Then, we apply this bound recursively to obtain final result. 
\end{proof}

\begin{proof} (Corollary~\ref{trade-off})
Denote the perturbation of each client $i$ as $\delta_i$. The perturbation will only influence the averaging step, formulated as 
$\bar{V}_{\delta} = \frac{1}{n}\sum_{i=1}^n (V_{(i)}+\delta_i)=\bar{V}+\bar{\delta},$
where $\bar{\delta}=\frac{1}{n}\sum_{i=1}^n\delta_i$ with mean $0$ and variance $2s^2/n$.
We know
$$\nabla\psi_{\delta}(x)=(V_{(1)}-\bar{V}-\bar{\delta},\cdots,V_{(n)}-\bar{V}-\bar{\delta})^\mathrm {T}$$
We can calculate the momentums of $\nabla\psi_{\delta}(x)$: 
$$\mathbb{E}[\nabla\psi_{\delta}(x)] = (V_{(1)}-\bar{V},\cdots,V_{(n)}-\bar{V})^\mathrm {T} = \nabla\psi(x),$$
$$\mathbb{E}[\Vert\nabla\psi_{\delta}(x)\Vert^2] =  \Vert\nabla\psi(x)\Vert^2+n\mathbb{E}[\bar{\delta}^2] = \Vert\nabla\psi(x)\Vert^2+2s^2.$$
Since $f(x)$ is not perturbed, we have 
$$\mathbb{E}[\Vert\nabla F_{\delta}(x)\Vert^2] = \mathbb{E}[\Vert\nabla f(x)+\lambda\nabla \psi_{\delta}(x)\Vert^2] = \Vert\nabla F(x)\Vert^2 + 2\lambda^2 s^2$$
Then we have 
\begin{align*}
\mathbb{E}[\Vert r^{k+1}\Vert]^2 &=\Vert r^k\Vert^2 -2\alpha\langle r^k,\nabla {F}(x^k)\rangle+\alpha^2\Vert\nabla{F}(x^k)\Vert^2 + 2\lambda^2 s^2\\
&\le(1-\alpha \mu)\Vert r^k\Vert^2 + 2\alpha^2\lambda^2 s^2 \quad(0\le\alpha\le\frac{1}{L+\lambda}),
\end{align*}
and we can use it recursively to obtain desired bound.
$$ \mathbb{E}[\Vert r^k\Vert^2]\le(1-{\alpha\mu})^k\Vert r^0\Vert^2+\frac{2\alpha\lambda^2 s^2}{\mu}. $$
Similarly, we can obtain the convergence result of RFRecF.
\end{proof}

\begin{acks}
This research was partially supported by Research Impact Fund (No.R1015-23), APRC - CityU New Research Initiatives (No.9610565, Start-up Grant for New Faculty of CityU), CityU - HKIDS Early Career Research Grant (No.9360163), Hong Kong ITC Innovation and Technology Fund Midstream Research Programme for Universities Project (No.ITS/034/22MS), Hong Kong Environmental and Conservation Fund (No. 88/2022), and SIRG - CityU Strategic Interdisciplinary Research Grant (No.7020046), Huawei (Huawei Innovation Research Program), Tencent (CCF-Tencent Open Fund, Tencent Rhino-Bird Focused Research Program), Ant Group (CCF-Ant Research Fund, Ant Group Research Fund), Alibaba (CCF-Alimama Tech Kangaroo Fund (No. 2024002)), CCF-BaiChuan-Ebtech Foundation Model Fund, Kuaishou, and Key Laboratory of Smart Education of Guangdong Higher Education Institutes, Jinan University (2022LSYS003).
\end{acks}


\clearpage
\bibliographystyle{ACM-Reference-Format}
\bibliography{reference}

\end{document}